\documentclass[envcountsame,envcountsect]{llncs}

\usepackage{microtype}
\usepackage{mathtools,amssymb,mathrsfs}
\usepackage[table,svgnames]{xcolor}
\usepackage[colorlinks=true,linkcolor=black,citecolor=MidnightBlue,urlcolor=MidnightBlue]{hyperref}
\usepackage{xspace}
\usepackage{tikz}
\usepackage{pifont}
\usepackage{adjustbox}

\usepackage{algorithm}
\usepackage{algorithmicx}
\usepackage{algpseudocode}
\newcommand{\negA}{\vspace{-0.05in}}
\newcommand{\negB}{\vspace{-0.1in}}

\newcommand{\mysection}[1]{\negB\section{#1}\negA}
\newcommand{\mysubsection}[1]{\negA\subsection{#1}\negA}

\newcommand{\myparagraph}[1]{\par\smallskip\par\noindent{\bf{}#1:~}}

\newcommand{\eps}{\varepsilon}

\newcommand{\comment}[1]{}

\begin{document}
\sloppy
\setcounter{page}{1} 
\title{Fast Distributed Approximation for Max-Cut}
\author{Keren Censor-Hillel\inst{1}\thanks{The research is supported in part by the Israel Science Foundation (grant 1696/14).} \and Rina Levy\inst{1} \and Hadas Shachnai\inst{1}}
\institute{Computer Science Department, Technion, Haifa 3200003, Israel. \mbox{E-mail: {\tt \{ckeren,rinalevy,hadas\}@cs.technion.ac.il.}}}

\maketitle

	
\begin{abstract}
Finding a maximum cut is a fundamental task in many computational settings. Surprisingly, it has been insufficiently studied in the classic distributed settings, where vertices communicate by 
synchronously sending messages to their neighbors according to the underlying graph, known as the $\mathcal{LOCAL}$ or $\mathcal{CONGEST}$ models. 
We amend this by obtaining 
almost optimal algorithms for Max-Cut on a wide class of  graphs in these models.
In particular, for any  $\epsilon > 0$, we develop randomized 
approximation algorithms achieving a ratio of $(1-\eps)$ to the optimum
 for Max-Cut  on bipartite graphs in the  
$\mathcal{CONGEST}$ model, and on {\em general} graphs in the 
$\mathcal{LOCAL}$ model.

We further present efficient {\em deterministic} algorithms, including a
$1/3$-approximation for Max-Dicut in 
our models, thus improving the best known (randomized) ratio of $1/4$.
Our algorithms make
non-trivial use of the greedy approach of Buchbinder et al. (\textit{SIAM
Journal on Computing, 2015}) for maximizing an unconstrained (non-monotone)
submodular function, which may be of independent interest.
\end{abstract}	



\mysection{Introduction}
Max-Cut is one of the fundamental problems in theoretical computer science. A {\em cut} in an undirected graph is a bipartition of the vertices, whose size is the number of edges crossing the bipartition. Finding cuts of maximum size in a given graph is among Karp's famous 21 NP-complete problems \cite{karp1972reducibility}. Since then, Max-Cut has received considerable attention,
in approximation algorithms \cite{sahni1976p,goemans1995improved,hadlock1975finding,trevisan2012max}, parallel computation~\cite{TanThesis}, parameterized complexity (see, e.g., \cite{saurabh2016k} and the references therein), and streaming algorithms (see, e.g., \cite{kapralov2015streaming}). 

Max-Cut has a central application in {\em wireless mesh networks (WMNs)}.
The capacity of WMNs
that operate over a single frequency can be increased significantly by enhancing each router with multiple {\em transmit (Tx)}
or {\em receive (Rx)} ({\em MTR}) capability. Thus, a node will not experience collision when two or more neighbors transmit to it. Yet, interference occurs if a node transmits and receives
simultaneously. This is known as the {\em no mix-tx-rx} constraint. The set of links activated in each time slot, defining the capacity of an MTR WMN,
 is governed by a link scheduler. As shown in~\cite{CSM12}, link scheduling 
is equivalent to finding Max-Cut in each time slot. A maximum cut contains the set of 
non-conflicting links that can be activated at the same time, i.e, they adhere to the no mix-tx-rx constraint. The induced bipartition of the vertices at each time slot defines
a set of transmitters and a set of receivers in this slot.  
Link scheduling algorithms based on approximating Max-Cut, and other applications in wireless networks, can be found
in~\cite{WL+13,XH+13,KB16,WCS16,XCRS16}.\footnote{Max-Cut naturally arises also in  VLSI \cite{chang1987efficient}, statistical physics \cite{barahona1988application} and machine learning \cite{wang2013semi}.}

Surprisingly, Max-Cut has been insufficiently studied in the classic distributed settings, where vertices communicate by synchronously sending messages to their neighbors according to the underlying graph, known as the $\mathcal{LOCAL}$ or $\mathcal{CONGEST}$ models.
Indeed, 
there are known distributed algorithms for Max-Cut using MapReduce techniques \cite{barbosa2015new,mirrokni2015randomized,mirzasoleiman2013distributed}. In this setting, the algorithms partition the ground set among $m$ machines and obtain a solution using all the outputs. However, despite a seemingly similar title, our distributed setting is completely different.

In this paper we address Max-Cut in the classic distributed network models, where the graph represents a synchronous communication network. At the end of the computation, each vertex decides locally whether it joins the subset $S$ or $\bar S$, and outputs $1$ or $0$, respectively, so as to obtain a cut of largest possible size.

It is well known that choosing a random cut, i.e., assigning each vertex to $S$ or $\bar{S}$ with probability $1/2$, yields a $\frac{1}{2}$-approximation for Max-Cut, and a $\frac{1}{4}$-approximation for Max-Dicut, defined on directed graphs (see, e.g., \cite{mitzenmacher2005probability,motwani2010randomized}).\footnote{In Max-Dicut we seek the maximum size edge-set crossing from 
$S$ to $\bar{S}$.}
Thus, a local algorithm, where each vertex outputs $0$ or $1$ with probability $1/2$, yields the above approximation factors with no communication required.
On the other hand, we note that a single vertex can find an optimal solution, once it has learned the underlying graph. However, this requires a number of communication rounds that depends {\em linearly} on global network parameters (depending on the exact model considered). This defines a tradeoff between time complexity and the approximation ratio obtained by distributed Max-Cut algorithms. The huge gap between the above results raises the following natural questions: How well can Max-Cut be approximated in the distributed setting, using a bounded number of communication rounds?
Or, more precisely: How many communication rounds are required for obtaining an approximation ratio strictly larger than half, or even a \emph{deterministic} $\frac{1}{2}$-approximation for Max-Cut?

To the best of our knowledge, these questions have been studied in our distributed network models only for a restricted graph class. Specifically, the paper \cite{hirvonen2014large} suggests a distributed algorithm for Max-Cut on $d$-regular triangle-free graphs, that requires a single communication round and provides a $(1/2+0.28125/\sqrt{d})$-approximation.

The key contribution of our paper is in developing two main techniques for approximating Max-Cut and Max-Dicut in distributed networks, with \emph{any} communication graph. Below we detail the challenges we face, and our methods for overcoming them.

\mysubsection{The Challenge}
In the $\mathcal{LOCAL}$ model, where message sizes and the local computation power are unlimited, every standard graph problem can be solved in $O(n)$ communication rounds. For Max-Cut it also holds that finding an optimal solution requires $\Omega(n)$ communication rounds. This lower bound follows from Linial's seminal lower bound \cite[Theorem 2.2]{linial1992locality} for finding a 2-coloring of an even-sized cycle. In an even cycle, the maximum cut contains all edges. Therefore, finding a Max-Cut is equivalent to finding a 2-coloring of the graph.

An approach that proved successful in many computational settings $-$ in 
 tackling hard problems $-$ is to relax the optimality requirement and settle for approximate solutions. Indeed, in the distributed setting, many approximation algorithms have been devised to overcome the costs of finding exact solutions (see, e.g., \cite{kuhn2016local,lotker2008improved,ghaffari2013distributed,kuhn2010distributed,henzinger2016deterministic,nanongkai2014distributed,aastrand2010fast,aastrand2009local,lenzen2013distributed,DBLP:conf/podc/Bar-YehudaCS16}, and the survey of Elkin \cite{elkin2004distributed}). Our work can be viewed as part of this general approach. However, we face crucial hurdles attempting to use the known sequential toolbox for approximating Max-Cut in the distributed setting.

As mentioned above, a $\frac{1}{2}$-approximation for Max-Cut can be obtained easily with no need for communication. While this holds in all of the above models, improving the ratio of $1/2$ is much more complicated.
In the sequential setting, an approximation factor strictly larger than $1/2$ was obtained in the mid-1990's using semidefinite programming~\cite{goemans1995improved} (see Section~\ref{subsection - related work}). 
Almost two decades later, the technique was applied by~\cite{TanThesis} to obtain a parallel randomized algorithm for Max-Cut, achieving a ratio of $(1-\epsilon) 0.878$ to the optimum, 
for any $\eps > 0$. Adapting this algorithm to our distributed setting seems non-trivial, as it
relies heavily on global computation.
Trying to apply other techniques, such as local search, unfortunately leads to 
linear running time, 
 because of the need to compare values of global solutions.

Another obstacle that lies ahead is the lack of {\em locality} in Max-Cut, due to strong dependency between the vertices. The existence of an edge in the cut depends on the assignment of both of its endpoints. This results in a chain of dependencies and raises the question whether cutting the chain can still guarantee a good approximation ratio.

\mysubsection{Our Contribution}
We develop two main techniques for approximating Max-Cut, as well as Max-Dicut. Our first technique relies on the crucial observation that the cut value is additive for edge-disjoint sets of vertices. 
Exploiting this property, we design \emph{clustering-based} algorithms, in which we decompose the graph into small-diameter clusters, find an optimal solution within each cluster, and prove that the remaining edges still allow the final solution to meet the desired approximation ratio. 
An essential component in our algorithms is efficient graph decomposition to such small-diameter clusters connected by few edges (also known as a {\em padded partition}), inspired by
a parallel algorithm of~\cite{miller2013parallel} (see also~\cite{elkin2016distributed,ElkinN17}).

For general graphs, this gives $(1-\epsilon)$-approximation algorithms for Max-Cut and Max-Dicut, requiring $O(\frac{\log n}{\epsilon})$ communication rounds in the $\mathcal{LOCAL}$ model. For the special case of a bipartite graph, we take advantage of the graph structure to obtain an improved
clustering-based algorithm, which does not require large messages. The algorithm achieves a $(1-\epsilon)$-approximation for Max-Cut in $O(\frac{\log n}{\epsilon})$ rounds, in the more restricted $\mathcal{CONGEST}$ model. 

For our second technique, we observe that the contribution of a specific vertex to the cut depends only on the vertex itself and its immediate neighbors. 
We leverage this fact to make multiple decisions in parallel by independent sets of vertices. We find such sets using distributed coloring algorithms. 
Our \emph{coloring-based} technique,
which makes non-trivial use of the greedy approach of~\cite{buchbinder2015tight}  for maximizing an unconstrained 
submodular function, yields
deterministic $\frac{1}{2}$-approximation and $\frac{1}{3}$-approximation algorithms for Max-Cut and Max-Dicut, respectively, and a randomized $\frac{1}{2}$-approximation algorithm for Max-Dicut. Each of these algorithms requires $\tilde{O}(\Delta +\log^*n)$ communication rounds in the $\mathcal{CONGEST}$ model, where $\Delta$ is the maximal degree of the graph, and $\tilde{O}$ ignores polylogarithmic factors in $\Delta$. 

Finally, we present $\mathcal{LOCAL}$ algorithms which combine both of our techniques. Applying the coloring-based technique to low-degree vertices, and the clustering-based technique to high-degree vertices, allows as to design faster deterministic algorithms with approximation ratios of $\frac{1}{2}$ and $\frac{1}{3}$ for Max-Cut and Max-Dicut, respectively, requiring $\min\{\tilde{O}(\Delta +\log^*n),O(\sqrt{n})\}$ communication rounds. 
Table \ref{results table} summarizes our results.
\vspace{-5mm}
\begin{table}
	\caption[]{A summary of our results.}
	\begin{adjustbox}{width=1\textwidth}
		\small
		\begin{tabular} {|c|c|c|c|c|c|c|}
			\hline
			\multicolumn{4}{|c|}{Algorithm Properties} & \multicolumn{2}{|c|}{Approximation Ratio} & {}\\
			\hline
			Rounds & Deterministic & Model & Graph & Max-Cut & Max-Dicut & {}\\
			\hline
			no communication & \ding{55} & $\mathcal{CONGEST}$ & any & 1/2 & 1/4 & folklore\\
			$O(\log n / \epsilon)$ & \ding{55} & $\mathcal{CONGEST}$ & bipartite & $1-\epsilon$ & $-$ & new \\
			$O(\log n / \epsilon)$ & \ding{55} & $\mathcal{LOCAL}$ & any & $1-\epsilon$ & $1-\epsilon$ & new\\
			$\tilde{O}(\Delta +\log^*n)$ & \ding{51} & $\mathcal{CONGEST}$ & any & $1/2$ & $1/3$ & new\\
			$\tilde{O}(\Delta +\log^*n)$ & \ding{55} & $\mathcal{CONGEST}$ & any & $1/2$ & $1/2$ & new\\
			$\min\{\tilde{O}(\Delta +\log^*n),O(\sqrt{n})\}$ & \ding{51} & $\mathcal{LOCAL}$ & any & $1/2$ & $1/3$ & new\\		
			\hline
		\end{tabular}
	\end{adjustbox}
	\label{results table}
\end{table}
\vspace{-5mm}

\mysubsection{Background and Related Work}\label{subsection - related work}
The weighted version of Max-Cut is one of Karp's NP-complete problems \cite{karp1972reducibility}. The unweighted version that we study here is also known to be NP-complete \cite{garey1976some}. 

While there are graph families, such as planar and bipartite graphs, in which a maximum cut can be found in polynomial time \cite{hadlock1975finding,grotschel1981weakly}, in general graphs, even approximating the problem is NP-hard. In the sequential setting, one cannot obtain an approximation ratio better than $\frac{16}{17}$ for Max-Cut, or an approximation ratio better than $\frac{12}{13}$ for Max-Dicut, unless $P=NP$ \cite{trevisan2000gadgets,haastad2001some}.

Choosing a random cut, i.e., assigning each vertex to $S$ or $\bar{S}$ with probability $1/2$, yields a $\frac{1}{2}$-approximation for Max-Cut, and $\frac{1}{4}$-approximation for Max-Dicut. In the sequential setting there are also deterministic algorithms yielding the above approximation ratios \cite{sahni1976p,papadimitriou1988optimization}.
For 20 years there was no progress in improving the $1/2$ constant in the approximation ratio for Max-Cut, until in 1995, Goemans and Williamson \cite{goemans1995improved} achieved the currently best known approximation ratio for Max-Cut, using semidefinite programming. They present a $0.878$-approximation algorithm, which is optimal assuming the Unique Game Conjecture holds \cite{khot2007optimal}. In the same paper, Goemans and Williamson also give a $0.796$-approximation algorithm for Max-Dicut. This ratio was improved later by Matuura et al. \cite{matuura20010}, to 0.863.
Using spectral techniques, a $0.53$-approximation algorithm for Max-Cut was given by Trevisan \cite{trevisan2012max}. 
In \cite{kale2010combinatorial} Kale and Seshadhri present a combinatorial approximation algorithm for Max-Cut using random walks, which gives a $(0.5+\delta)$-approximation, where $\delta$ is some positive constant which appears also in 
the running time of the algorithm. In particular, for $\tilde{O}(n^{1.6}),\tilde{O}(n^2)$ and $\tilde{O}(n^3)$ times, the algorithm achieves approximation factors of $0.5051, 0.5155$ and $0.5727$, respectively.

Max-Cut and Max-Dicut can also be viewed as special cases of submodular maximization, which has been widely studied. It is known that choosing  a solution set $S$ uniformly at random yields a $\frac{1}{4}$-approximation, and a $\frac{1}{2}$-approximation for a general and for symmetric submodular function, respectively \cite{feige2011maximizing}. This corresponds to the known random approximation ratios for Max-Cut and Max-Dicut. Buchbinder et al. \cite{buchbinder2015tight} present determinstic $\frac{1}{2}$-approximation algorithms for both symmetric and asymmetric submodular functions.
These algorithms assume 
that the submodular function is accessible through a black box returning $f(S)$ for any given set $S$ (known as the {\em value oracle} model).

In the recent years, there is an ongoing effort to develop distributed algorithms for submodular maximization problems, using MapReduce techniques \cite{barbosa2015new,mirrokni2015randomized,mirzasoleiman2013distributed}. Often, the inputs consist of large data sets, for which a sequential algorithm may be inefficient. The main idea behind these algorithms is to partition the ground set among $m$ machines, and have each machine solve the problem optimally independently of others. After all machines have completed their computations, they share their solutions. A final solution is obtained by solving the problem once again over a union of the partial solutions. The algorithms achieve performance guarantees close to the sequential algorithms while decreasing the running time, where the running time is the number of communication rounds among the machines. As mentioned above, these algorithms do not apply to our classic distributed settings.


\mysection{Preliminaries}
The Max-Cut problem is defined as follows. Given an undirected graph $G=(V,E)$, one needs to divide the vertices into two subsets, $S\subset{V}$ and $\bar{S}=V\setminus{S}$, such that the size of the cut, i.e., the number of edges between $S$ and the complementary subset $\bar{S}$, is as large as possible.
In the Max-Dicut problem, the given graph $G=(V,E)$ is directed, and the cut is defined only as the edges which are directed from $S$ to $\bar{S}$. As in the Max-Cut problem, the goal is to obtain the largest cut.

Max-Cut and Max-Dicut can be described as the problem of maximizing the submodular function $f(S)=|E(S,\bar{S})|$, where for Max-Dicut $f(S)$ counts only the edges directed from $S$ to $\bar{S}$.
Given a finite set $X$, a \emph{submodular} function is a function $f:2^X\to\mathbb{R}$, where $2^X$ denotes the power set of $X$, which satisfies the equivalent definitions:
\begin{enumerate}
	\item For any $S,T\subseteq{X}$:  $f(S\cup{T})+f(S\cap{T})\le{f(S)+f(T)}.$
	\item For any $A\subseteq{B}\subseteq{X}$ and $x\in{X}\setminus{B}$:  $f(B\cup{\{x\}})-f(B)\le{f(A\cup{\{x\}})-f(A)}.$
\end{enumerate}

For Max-Cut and Max-Dicut, the submodular function also satisfies the following equality:
For every disjoint sets $S,T\subseteq X$ such that $E_{S\times T}=\{(u,v)|u\in S, v\in T\}=\emptyset$, we have that $f(S)+f(T)=f(S\cup T)$.
Note that for Max-Cut, the function is also symmetric, i.e., $f(S)=f(\bar{S})$.

\myparagraph{Model}\label{model} 
We consider a distributed system, modeled by a graph $G=(V,E)$, in which the vertices represent the computational entities, and the edges represent the communication channels between them. We assume that each vertex $v$ has a unique identifier $id(v)$ of size $O(\log{n})$, where $n=|V|$.

The communication between the entities is synchronous, i.e., the time is divided into rounds. In each round, the vertices send messages simultaneously to all of their neighbors and make a local computation based on the information gained so far. This is the classic $\mathcal{LOCAL}$ model \cite{peleg2000distributed}, which focuses on analyzing how locality affects the distributed computation. Therefore, the messages size and local computation are unlimited, and the complexity is measured by the number of communication rounds needed to obtain a solution. It is also important to study what can be done in the more restricted $\mathcal{CONGEST}$ model \cite{peleg2000distributed}, in which the message size is bounded to $O(\log n)$.

We assume that each vertex has preliminary information including the size of the network $n=|V|$, its neighbors, and the maximal degree of the graph
 $\Delta$.\footnote{This assumption is needed only for the $(\Delta +1)$-coloring algorithm \cite{barenboim2015deterministic} used in Section \ref{Section - coloring based algorithms}; it can be omitted
 (see \cite{barenboim2015deterministic}), increasing the running time by a constant factor.}

Each vertex runs a local algorithm in order to solve the Max-Cut problem. Along the algorithm, each vertex decides locally whether it should be in $S$ or in $\bar S$, and outputs $1$ or $0$ respectively. We define the \emph{solution} of the algorithm, as the set of all outputs. Note that each vertex does not hold the entire solution, only local information. The solution $value$ is defined as the size of the cut induced by the solution. We show that this value approximates the size of the maximum cut.


\mysection{Clustering-Based Algorithms} 
In this section we present clustering-based algorithms for Max-Cut and Max-Dicut. Our technique uses the observation that Max-Cut is a collection of edges having their endpoints in different sets; therefore, it can be viewed as the union of cuts in the disjoint parts of the graph.

Given a graph $G = (V,E)$, we first eliminate a small fraction of edges to obtain small-diameter connected components. Then, the problem is solved optimally within each connected component. For general graphs, this is done by gathering the topology of the component at a single vertex. For the special case of a bipartite graph, we can use the graph structure to propagate less information.
Since the final solution, consisting of all the vertices local decisions, is at least as good as the sum of the optimal solutions in the components, and since the fraction of eliminated edges is small, we prove that the technique yields a $(1-\epsilon)$-approximation.

\mysubsection{A Randomized Distributed Graph Decomposition}
We start by presenting the randomized distributed graph decomposition algorithm. The algorithm is inspired by a parallel graph decomposition by Miller et al. \cite{miller2013parallel} that we adapt to the distributed model as we describe next.\footnote{Our algorithm can be viewed as one phase of the distributed algorithm presented by Elkin et al. in \cite{elkin2016distributed} with some necessary changes.}. The PRAM algorithm of~\cite{miller2013parallel} generates a {\em strong padded partition} of a given graph, namely, a partition into connected
 components with strong diameter $O(\frac{\log n}{\beta})$, for some $\beta \leq 1/2$, such that the fraction of edges that cross between different clusters of the partition is 
 at most $\beta$.
 As we prove next, the distributed version guarantees the same properties with high probability and requires only $O(\frac{\log n}{\beta})$ communication rounds in the $\mathcal{CONGEST}$ model.

The distributed version of the graph decomposition algorithm works as follows: Let $\delta_v$ be a random value that vertex $v$ chooses from an exponential distribution with parameter $\beta$. Define the \emph{shifted distance} from vertex $v$ to vertex $u$ as $dist_{\delta} (u,v) = dist(u,v)-\delta_u$. Along the algorithm each vertex $v$ finds a vertex $u$ within its $\frac{k\log n}{\beta}$-neighborhood, where $k$ is a constant, that minimizes $dist_{\delta} (u,v)$. We define this vertex as $v$'s \emph{center}. This step creates the difference between the parallel and the distributed decomposition, as in the parallel algorithm each vertex chooses its center from the entire ground set $V$. However as we prove next, the process still generates a decomposition with the desired properties. Furthermore, w.h.p. the distributed algorithm outputs a decomposition identical to the one created by the parallel algorithm.
A pseudocode of the algorithm is given in Algorithm \ref{alg - partition}.
\begin{algorithm}
	\caption{Distributed Decomposition, \emph{code for vertex $v$}}
	\label{alg - partition}
	\begin{algorithmic}[1]
		\State{$0<\beta<1, k>2.$}
		\State{choose $\delta_v$ at random from $Exp(\beta)$}
		\State $center=id(v)$
		\State $dist_{\delta_{min}} = -\delta_v$
		\For{$\frac{k\log n}{\beta}$ iterations}
		\State send $(dist_{\delta_{min}},center)$
		\For {every $(dist_{\delta_{min}}^{'} ,center^{'})$ received from $u\in N(v)$}
		\If  {$\big(dist_{\delta_{min}}^{'} + 1 < dist_{\delta_{min}}\big)$ OR $\big((dist_{\delta_{min}}^{'} + 1 = dist_{\delta_{min}})$ AND $(center^{'}<center)\big)$}
		\State $center \leftarrow center^{'}$ 
		\State $dist_{\delta_{min}}\leftarrow dist_{\delta_{min}}^{'}+1$
		\EndIf
		\EndFor
		\EndFor
		\State{output $center$}
	\end{algorithmic}
\end{algorithm}
We prove that the fraction of edges between different components is small. In order to do so, we bound the probability of an edge to be between components, i.e., the probability that the endpoints of the edge choose different centers.
We consider two cases for an edge $e=(u,v)$. In the first case, we assume that both $u$ and $v$ choose the center that minimizes their shifted distance, $dist_\delta$, over all the vertices in the graph. In other words, if the algorithm allowed each vertex to learn the entire graph, they would choose the same center as they did in the current algorithm. In the second case, we assume that at least one of $u$ and $v$ chooses differently if given a larger neighborhood.

Define the \emph{ideal} center of a vertex $v$ as $argmin_{w\in V}dist_\delta(w,v)$. In the following lemma we bound from above the probability that a vertex does not choose its ideal center.
\negA
\begin{lemma}\label{center distance}
	Let $v'$ be the ideal center of vertex $v$, then the probability that $dist(v',v)>\frac{k\log n}{\beta}$, i.e., vertex $v$ does not join its ideal center, is at most $\frac{1}{n^k}$.
\end{lemma}

\begin{proof}
	Since $v'$ is the ideal center of vertex $v$, we have that $dist_\delta(v',v) \leq dist_\delta(v,v)$. Therefore, $dist(v',v)-\delta_{v'} \leq dist(v,v)-\delta_v = -\delta_v \leq 0,$ which implies that $dist(v',v) \leq \delta_{v'}$. That is, the distance between each vertex $v$ to its ideal center $v'$ is bounded from above by $\delta_{v'}$, and hence $\Pr\left[dist(v',v)>\frac{k\log n}{\beta}\right]\leq\Pr\left[\delta_{v'}>\frac{k\log n}{\beta}\right]$.
	Using the cumulative exponential distribution, we have that
	$\Pr\left[\delta_{v'}>\frac{k\log n}{\beta}\right]=\exp\left(-\frac{k\cdot \beta\log n}{\beta}\right)=\exp\left(-k\log n\right) \ \leq \frac{1}{n^k}$.
	\hfill \qed
\end{proof}

\begin{corollary}\label{corollary - distributed equals parallel w.h.p}
	The Distributed Decomposition algorithm generates a decomposition identical to the decomposition created by the parallel decomposition algorithm with probability at least $1-\frac{1}{n^{k-1}}$
\end{corollary}
\comment{
\begin{proof}
	Applying the union bound, 
	we have that the probability that at least one of the vertices does not choose its ideal center is at most $\frac{1}{n^{k-1}}$.\qed
\end{proof}
}
Define an \emph{exterior} edge as an edge connecting different vertex components, and let $F$ denote the set of exterior edges. Let $A_{u,v}$ denote the event that both $u$ and $v$ choose their ideal centers.

\begin{lemma}\label{beta bound}
	The probability that an edge $e=(u,v)$ is an exterior edge, given that $u$ and $v$ choose their ideal centers, is at most $\beta$.
\end{lemma}

The lemma follows directly from \cite{miller2013parallel}, where indeed the algorithm assigns to each vertex its ideal center.
We can now bound the probability of any edge to be an exterior edge.

\begin{lemma}\label{crossing edge probability}
	The probability that an edge $e=(u,v)$ is in $F$ is at most $\beta+\frac{2}{n^k}$.
\end{lemma}

\begin{proof}
%
		Note that $$\Pr\left[(u,v)\in F\right] = \Pr\left[(u,v)\in F\big|A_{u,v}\right]\Pr\left[A_{u,v}\right]+\Pr\left[(u,v)\in F\big|\bar A_{u,v}\right]\Pr\left[\bar A_{u,v}\right].$$ By Lemma \ref{beta bound}, $\Pr\left[(u,v)\in F\big|A_{u,v}\right]\leq\beta$. Applying the union bound on the result of Lemma \ref{center distance}, we have that $\Pr\left[\bar A_{u,v}\right] \leq \frac{2}{n^k}$.
	Therefore $\Pr\left[(u,v)\in F\right] = \Pr\left[(u,v)\in F\big|A_{u,v}\right] \Pr\left[A_{u,v}\right]+\Pr\left[(u,v)\in F\big|\bar A_{u,v}\right]\Pr\left[\bar A_{u,v}\right]\leq \beta\cdot \Pr\left[A_{u,v}\right]+\Pr\left[(u,v)\in F\big|\bar A_{u,v}\right]\cdot \frac{2}{n^k} \leq \beta + \frac{2}{n^k}.$
\end{proof}

We can now prove the guarantees of the Distributed Decomposition algorithm. 
Recall that the \emph{weak diameter} of a set $S=\{u_1,u_2,...u_l\}$ is defined as $\max_{(u_i,u_j)\in S} dist(u_i,u_j)$.

\begin{theorem} \label{Partition}
	The Distributed Decomposition algorithm requires $O(\frac{\log n}{\beta})$ communication rounds in the $\mathcal{CONGEST}$ model, and partitions the graph into components such that in expectation there are $O(\beta m)$ exterior edges. Each of the component is of weak diameter $O(\frac{\log n}{\beta})$, and with high probability also of strong diameter $O(\frac{\log n}{\beta})$.
\end{theorem}

\begin{proof}
	Clearly, as every vertex chooses a center from its $\frac{k\log n}{\beta}$-neighborhood, the distance between two vertices that choose the same center, i.e., belong to the same component, over the graph $G$ is at most $O(\frac{\log n}{\beta})$. Therefore, the weak diameter of every component is at most $O(\frac{\log n}{\beta})$. 
	As we proved in Corollary \ref{corollary - distributed equals parallel w.h.p}, with probability at least $1-\frac{1}{n^{k-1}}$ the algorithm creates a partition identical to the one created by the parallel algorithm, and therefore with the exact same properties, which implies that the strong diameter of every component is at most $O(\frac{\log n}{\beta})$ as well.
	
	Using the linearity of expectation, and Lemma \ref{crossing edge probability} we have that $\mathbb{E}\left[|F|\right]\leq\sum_{e\in E}^{}\left(\beta+\frac{2}{n^k}\right) = \beta m+\frac{2m}{n^k}.$
	Since $m\leq n^2$, we have that for every $k>2$, $\mathbb{E}\left[|F|\right] \leq O(\beta m)$.
	Finally, as can be seen from the code, the algorithm requires $O(\frac{\log n}{\beta})$ communication rounds.\qed
\end{proof}

\mysubsection{A Randomized $(1-\epsilon)$-Approximation Algorithm for Max-Cut on a Bipartite Graph}
Clearly, in a bipartite graph the maximum cut contains all of the edges. Such a cut can be found by selecting arbitrarily a root vertex, and then simply putting all the vertices of odd depth in one set and all the vertices of even depth in the complementary set. However, this would require a large computational time in our model, that depends on the diameter of the graph. We overcome this by using the above decomposition,
and finding an optimal solution within each connected component. In each component $C$, we find an optimal solution in $O(D_c)$ communication rounds, where $D_c$ is the diameter of $C$.
First, the vertices in each component search for the vertex with the lowest id. \footnote{This can be done 
by running a BFS in parallel from all vertices.
Each vertex propagates the information from the root with lowest id it knows so far, and joins its tree. Thus, at the end of the process, we have a BFS tree rooted at the vertex with the lowest id.}Second, the vertex with the lowest id joins $S$ or $\bar{S}$ with equal probability and sends its decision to its neighbors. When a vertex receives a message from one of its neighbors, it joins the opposite set, outputs its decision, and sends it to its neighbors.
Since finding the optimal solution within each component does not require learning the entire component topology, the algorithm is applicable to the more restricted $\mathcal{CONGEST}$ model.
The algorithm yields a $(1-\epsilon)$-approximation for the Max-Cut problem on a bipartite graph in $O(\frac{\log n}{\epsilon})$ communication rounds with high probability.

\begin{theorem}\label{theorem - bipartite max-cut}
	Bipartite Max-Cut is a randomized $(1-\epsilon)$-approximation for Max-Cut, requiring $O(\frac{\log n}{\epsilon})$ communication rounds in the $\mathcal{CONGEST}$ model w.h.p.
\end{theorem}
\begin{proof}
	After applying the decomposition algorithm, we have that each connected component $C$ has a diameter $D_c$ of at most $O(\frac{\log n}{\beta})$  w.h.p (Theorem \ref{Partition}). Building a BFS tree in a component $C$ clearly takes $O(D_c)$ communication rounds. Assigning the vertices to sets after constructing a tree takes $O(D_c)$ as well. Therefore, the algorithm finds an optimal solution in each of the components in $O(\frac{\log n}{\epsilon})$ communication rounds. Since every connected component is a bipartite graph itself, all the edges within it are in the cut. Therefore, as there are at most $O(\beta m)$ exterior edges, the algorithm obtain a cut of at least $(1-\beta)m$ edges. Since the optimal cut in a bipartite graph contains all the edges, the algorithm achieves a $(1-\beta)$-approximation. Choosing $\beta = \epsilon$ proves the theorem.\qed
\end{proof}

\negA
\begin{algorithm}
	\caption{Bipartite Max-Cut} 
	\label{alg1} 
	\begin{algorithmic}[1]
		\State{G=(V,E)}
		\State{apply Distributed Decomposition on G, with $\beta=\epsilon,k>2$}
		\For{each component $C$ obtained by the decomposition}
		\State{build a BFS tree from the vertex $v$ with the lowest id}
		\State{assign $v$ to $S$ or $\bar{S}$ with equal probability, assign the rest of the vertices to alternating sides}
		\EndFor
	\end{algorithmic}
\end{algorithm}
\negA
\mysubsection{A Randomized $(1-\epsilon)$-Approximation Algorithm for General Graphs} 
We present below a $(1-\epsilon)$-approximation algorithm for Max-Cut in general graphs, using $O(\frac{\log n}{\epsilon})$ communication rounds. As before, the algorithm consists of two parts, decomposition and solution. Although the decomposition part itself works even in the $\mathcal{CONGEST}$ model, the algorithm works in the $\mathcal{LOCAL}$ model, since for general graphs the components created by the decomposition are not necessarily sparse, and learning the components topology is expensive in the $\mathcal{CONGEST}$ model. 

\begin{algorithm} 
	\caption{Decomposition-Based Max-Cut}
	\label{alg 1-beta}
	\begin{algorithmic}[1]
		\State{G=(V,E)}
		\State{apply Distributed Decomposition on G, with $\beta=\epsilon/2,k>2$}
		\For{each component $C$ obtained by the decomposition}
		\State{gather the component topology at the vertex $v\in C$ with the lowest id.}
		\State{let $v$ find an optimal solution and determine the value output by the component's vertices.}
		\EndFor
	\end{algorithmic}
\end{algorithm}
\negA
\begin{theorem}
	Decomposition-Based Max-Cut is a randomized $(1-\epsilon)$-approximation for Max-Cut, requiring $O(\frac{\log n}{\epsilon})$ communication rounds in the $\mathcal{LOCAL}$ model.
\end{theorem}

\begin{proof}
	Let $\mathit{OPT}(G)$ be the set of edges that belong to some maximum cut in $G$, and let $ALG(G)$ be the set of edges in the cut obtained by Decomposition-Based Max-Cut. Let $S_u$ be the component induced by the vertices which choose $u$ as their center, and denote by $S$ the set of components that algorithm Distributed Decomposition constructs. Then $\mathbb{E}\left[|ALG(G)|\right]\geq \mathbb{E} \left[\sum_{S_u\in S}|\mathit{OPT}(S_u)|\right] \geq |\mathit{OPT}(G)|-\beta m \geq |\mathit{OPT}(G)|-2\beta |\mathit{OPT}(G)|=(1-\epsilon)|\mathit{OPT}(G)|.$
	The last inequality follows from the fact that for every graph $G$ it holds that $|\mathit{OPT}(G)|\geq \frac{m}{2}$. 
	
	The graph decomposition requires $O(\frac{\log n}{\epsilon})$ communication rounds, and outputs components with weak diameter at most $O(\frac{\log n}{\epsilon})$. Therefore, finding the optimal solution within each component takes $O(\frac{\log n}{\epsilon})$ as well. The time bound follows.\qed
\end{proof}

By taking $\beta=\epsilon/4$, one can now obtain a $(1-\epsilon)$-approximation algorithm for Max-Dicut. The difference comes from the fact that for Max-Dicut it holds that $|\mathit{OPT}(G)|\geq \frac{m}{4}$ for every graph $G$. The rest of the analysis is similar to the analysis for Max-Cut. Hence, we have
\begin{theorem}
	Decomposition-Based Max-Dicut is a randomized $(1-\epsilon)$-approximation for Max-Dicut, requiring $O(\frac{\log n}{\epsilon})$ communication rounds in the $\mathcal{LOCAL}$ model.
\end{theorem} 



\mysection{Coloring-Based Algorithms}\label{Section - coloring based algorithms}
Many of the sequential approximation algorithms for Max-Cut perform $n$ iterations. Each vertex, in its turn, makes a greedy decision so as to maximize the solution value.
We present below distributed greedy algorithms which achieve the approximation ratios of the sequential algorithms much faster.
We first prove that the greedy decisions of vertices can be done locally, depending only on their immediate neighbors. Then we show how to parallelize the decision process, such that in each iteration an independent set of vertices completes.
The independent sets are generated using $(\Delta +1)$-coloring; then, for $(\Delta +1)$ iterations, all the vertices of the relevant color make their parallel independent decisions.
All algorithms 
 run in the $\mathcal{CONGEST}$ model.


\subsection{A Deterministic 1/2-Approximation Algorithm for Max-Cut}\label{Subsection - deterministic 1/2 Max-Cut}
We start by presenting a simple deterministic distributed greedy algorithm that yields a 1/2-approximation for Max-Cut.
The algorithm is inspired by the sequential greedy Max-Cut algorithm of \cite{sahni1976p}. The greedy Max-Cut algorithm iterates through the vertices in some arbitrary order. In iteration $i$, the algorithm decides whether to assign vertex $v_i$ to $S$ or to $\bar S$, based on which placement maximizes the cut size.
In our algorithm the process is similar, except that, in each iteration, instead of considering a single vertex, we consider an independent set of vertices. As the vertices are independent, their decisions are also independent, and the approximation ratio still holds.

To divide the vertices into independent sets, we color them using $(\Delta +1)$-colors, where $\Delta$ is the maximum degree in the graph. The best deterministic $(\Delta +1)$-coloring algorithm known in the $\mathcal{CONGEST}$ model, due to Barenboim \cite{barenboim2015deterministic}, requires $\tilde{O}(\Delta^{3/4}+\log^* n)$ communication rounds, where $\tilde{O}$ ignores polylogarithmic factors in $\Delta$ \footnote{Note that a faster $(\Delta + 1)$-coloring algorithm will not improve the running time of Distributed Greedy Max-Cut, since the running time of the algorithm depends on the number of colors, which is $\Delta + 1$.}. 

Define a coloring $C:V\rightarrow\{1,2,...,\Delta +1\}$ such that $C(v)\neq C(u)$ for every $(u,v)\in E$. Let $N_{low}(v)=\{u\mid u\in N(v)$ and $C(u)<C(v)\}$ denote the neighbors of vertex $v$ with lower colors.
In iteration $i$, all vertices with color $i$ decide in parallel whether to join $S$ or $\bar S$, depending on the choices their neighbors made in previous rounds. In order to maximize the cut locally, vertex $v$ chooses to join the subset that was chosen by the minority of its neighbors in $N_{low}(v)$.
As we show next, this guarantees the 1/2-approximation.

Algorithm \ref{alg-GreedyMaxCut} gives a pseudocode of the algorithm.

\begin{algorithm}
	\caption{Distributed Greedy Max-Cut} 
	\label{alg-GreedyMaxCut}
	\begin{algorithmic}[1]
		\State color the graph using $(\Delta +1)$ colors
		\For{i=1 to $(\Delta +1)$}
		\If{$C(v)=i$}
		\If{$|N_{low}(v)\cap S|\leq |N_{low}(v)\cap \bar S|$}
		\State{State($v$) = $S$}
		\State{send 1}
		\Else
		\State{State($v$) = $\bar S$}
		\State{send 0}
		\EndIf
		\EndIf
		\EndFor
		\State{output State}
	\end{algorithmic}
\end{algorithm}

\begin{theorem} \label{Theorem - 1/2 deterministic Max-Cut }
	The Distributed Greedy Max-Cut algorithm outputs a $\frac{1}{2}$-approximation in $\tilde{O}(\Delta + \log^* n)$ rounds.
\end{theorem}

\begin{proof}
	We first show that the algorithm gives a $\frac{1}{2}$-approximation. Consider an edge $e=(u,v)$, if $C(v)>C(u)$, then we say that $v$ is $responsible$ for edge $e$. Denote by $E_{resp}(v)$ the set of edges vertex $v$ is responsible for. In other words, $E_{resp}(v)$ contains the edges between $v$ and vertices in $N_{low}(v)$. Since the color groups are independent, every edge has exactly one responsible vertex, therefore $\sum_{v\in V}^{} |E_{resp}(v)| = |E|$. When vertex $v$ makes its decision, it chooses to join the set that is not chosen by at least half of its neighbors in $N_{low}(v)$, and hence adds at least half of the edges in $E_{resp}(v)$ to the cut. Upon termination of the algorithm, we have that $|E(S,\bar S)|\geq\sum_{v\in V}^{} \frac{1}{2}|E_{resp}(v)|=\frac{1}{2}|E|$. Since the size of the optimal cut cannot be larger than $|E|$, the algorithm yields a $\frac{1}{2}$-approximation.
	
	The algorithm colors the graph in $\tilde{O}(\Delta^{3/4}+\log^* n)$ rounds and iterates for $O(\Delta + 1)$ rounds, which yields the time bound.\qed
\end{proof}

\subsection{A Deterministic 1/3-Approximation Algorithm for Max-Dicut}
Next, we turn our attention to the Max-Dicut problem. Buchbinder et al. \cite{buchbinder2015tight} present a sequential deterministic greedy 1/3-approximation algorithm for maximizing unconstrained submodular functions which runs in linear time. Inspired by this algorithm, we present a distributed deterministic 1/3-approximation algorithm for Max-Dicut that requires $\tilde{O}(\Delta + \log^*n)$ communication rounds in the $\mathcal{CONGEST}$ model.

\subsubsection{The sequential Algorithm}
We first give a brief overview of the sequential algorithm of \cite{buchbinder2015tight} for Max-Dicut. Given a graph $G=(V,E)$, where $|V|=n$, the algorithm examines the vertices in an arbitrary order ${u_1,u_2,...,u_n}$. In iteration $i$, the algorithm decides greedily whether to include $u_i$ in the final solution, for $1\leq i \leq n$. The algorithm maintains two sets of vertices, $X$ and $Y$. Initially, $X_0=\emptyset$ and $Y_0=V$. At the begining of the $i$th iteration, the algorithm defines $X_i=X_{i-1},Y_i=Y_{i-1}$, then, the algorithm decides whether to add the $i$th vertex $u_i$ to $X_{i-1}$, or to remove it from $Y_{i-1}$. The decision is made by calculating the marginal profit of both options and choosing the more profitable one.
By definition, upon termination of the algorithm we have that $X_n=Y_n$, and this set is output as the solution.
Let $f(S)$ be the size of the directed cut induced by a subset of vertices $S\subseteq V$, i.e., the number of edges directed from $S$ to $\bar{S}$. Then $f$ is a non-negative submodular function.
Let $a_i$ and $b_i$ denote the marginal profit gained by adding the vertex $u_i$ to $X_{i-1}$, or removing it from $Y_{i-1}$ respectively.
Algorithm \ref{Alg-buchbinder} gives the pseudocode of the sequential algorithm of \cite{buchbinder2015tight}.

\begin{algorithm}
	\caption{Deterministic Sequential Max-Dicut}
	\label{Alg-buchbinder}
	\begin{algorithmic}[1]
		\State{$G=(V,E)$}
		\State $X_0 \leftarrow \emptyset, Y_0 \leftarrow V$
		\For{i=1 to n}
		\State $a_i \leftarrow f(X_{i-1}+u_i)-f(X_{i-1})$
		\State $b_i \leftarrow f(Y_{i-1}-u_i)-f(Y_{i-1})$
		\If{$a_i \geq b_i$}
		\State{$X_i\leftarrow X_{i-1}+u_i$}
		\State{$Y_i \leftarrow Y_{i-1}$}
		\Else
		\State{$X_i\leftarrow X_{i-1}$}
		\State{$Y_i \leftarrow Y_{i-1}-u_i$}
		\EndIf
		\EndFor
	\end{algorithmic}
\end{algorithm}

We give below a sketch of the analysis \footnote{See the details in \cite{buchbinder2015tight}.}. The following lemma implies that in each iteration, one can only increase the value of the solution.

\begin{lemma}\label{lemma 1 -usm}\cite[Lemma 2.1]{buchbinder2015tight}
	For every $1\leq i \leq n$, it holds that $a_i+b_i \geq 0$.
\end{lemma}

Let $\mathit{OPT}$ denote the set $S\subseteq V$ that maximizes the dicut size. Define $\mathit{OPT}_i \triangleq (\mathit{OPT}\cup X_i)\cap Y_i$. In other words, $\mathit{OPT}_i$ agrees with $X_i$ and $Y_i$ on the first $i$ elements, and agrees with $\mathit{OPT}$ on the rest. Hence, $\mathit{OPT}_0 = \mathit{OPT}$ and $\mathit{OPT}_n = Y_n = X_n$. The following lemma shows that in each iteration, the damage to the optimal solution value, i.e., $f(\mathit{OPT}_{i-1})-f(\mathit{OPT}_i)$, is bounded.

\begin{lemma}\label{lemma 2-usm}\cite[Lemma 2.2]{buchbinder2015tight}
	For every $1\leq i \leq n$, it holds that $f(\mathit{OPT}_{i-1})-f(\mathit{OPT}_i) \leq [f(X_i)-f(X_{i-1})] + [f(Y_i)-f(Y_{i-1})].$
\end{lemma}

Using Lemmas \ref{lemma 1 -usm},\ref{lemma 2-usm}, one can prove the following theorem.

\begin{theorem}[following from \cite{buchbinder2015tight}]
	The Deterministic Sequential Max-Dicut algorithm gives a 1/3-approximation for the Max-Dicut problem in linear time.
\end{theorem}

\subsubsection{The Distributed Algorithm}
Inspired by the sequential algorithm, we design a distributed algorithm which gives a 1/3-approximation for the Max-Dicut problem.
As in Subsection \ref{Subsection - deterministic 1/2 Max-Cut}, we start by ($\Delta +1$)-coloring the graph. Then, for $1 \leq i \leq (\Delta +1)$ iterations, the vertices make their decisions, one color class at a time. In each iteration, the corresponding vertices calculate the marginal profit gained by their two possible decisions, and take the one which maximizes the profit.
As in the sequential algorithm, the distributed algorithm maintains two solutions $X$ and $Y$; $X_0=\emptyset$ and $Y_0 = V$ as before. $X_i$ and $Y_i$ represent the state of the solutions after $i$ iterations, and $\mathit{OPT}_i$ is defined as $\mathit{OPT}_i=(\mathit{OPT}\cup X_i)\cap Y_i)$.

Define $X_i(v) \triangleq X_i \cap N(v)$, and $Y_i(v) \triangleq Y_i \cap N(v)$. It is easy to see that $X_i$ as defined in the sequential algorithm equals $\cup_{v\in V} X_i(v)$. Similarly, $Y_i = \cup_{v\in V} Y_i(v)$.
Using this notation, Algorithm \ref{alg3} gives the pseudocode of the distributed algorithm.

\begin{algorithm}[H] 
	\caption{Distributed Greedy Max-Dicut}
	\label{alg3}
	\begin{algorithmic}[1]
		\State $X_0(v)=\emptyset, Y_0(v)=N(v)$
		\State color the graph using $(\Delta +1)$ colors
		\For{i=1 to $(\Delta +1)$}
		\If{$C(v)=i$} 
		\State $a_i \leftarrow f(X_{i-1}(v)+v)-f(X_{i-1}(v))$
		\State $b_i \leftarrow f(Y_{i-1}(v)-v)-f(Y_{i-1}(v))$
		\If{$a_i\geq b_i$}
		\State{State($v$) = $S$}
		\State{send 1}
		\Else
		\State{State($v$) = $\bar S$}
		\State{send 0}
		\EndIf
		\EndIf
		\For{each vertex $v$ in $V$ in parallel}
		\State{$X_i(v) \leftarrow X_{i-1}(v) + \{u\mid u\in N(v)$ $\land$ $u$ sent 1$\}$}
		\State{$Y_i(v) \leftarrow Y_{i-1}(v) - \{u\mid u\in N(v)$ $\land$ $u$ sent 0$\}$}
		\EndFor
		\EndFor
		\State output State
	\end{algorithmic}
\end{algorithm}

There are two key ingredients in our analysis. We first prove that the marginal profits $a_i$ and $b_i$ can be computed locally. Then, we need to show that running the procedures in parallel does not affect the approximation ratio.

The next lemma shows that the marginal profits of $v$'s possible decisions depends only on its 1-neighborhood.

\begin{lemma}\label{local marginal profit}
	Let $A\subseteq V$ be a subset of vertices, and let $v$ be a vertex such that $v\notin A$. Then $f(A+v)-f(A)=f((A\cap N(v)) + v)-f(A\cap N(v))$.
\end{lemma}

\begin{proof}
	Given $A,B\subseteq V$ such that $A\cap B = \emptyset$, let $|E(A,B)|$ denote the number of edges directed from vertices in $A$ to vertices in $B$.
	We start by proving that for every subset $A\subseteq V$ and $v\notin A$ it holds that:
	\begin{equation*}
	\tag{$\star\star$}
	\begin{split}
	&f(A + v) - f(A) = |E(A+v,V\setminus(A+v))|-|E(A,V\setminus A)|\\
	&=|E(A,V\setminus(A+v))|+|E(v,V\setminus(A+v))|-|E(A,V\setminus(A+v))|-|E(A,v)|\\
	&=|E(v,V\setminus(A+v))|-|E(A,v)|.
	\end{split}		
	\end{equation*}
	
	Note that since $v$ is connected by edges only to its neighbors, $$|E(v,V\setminus(A+v))|-|E(A,v)| = \left|E\left(v,N(v) \cap \left(V\setminus(A+v)\right)\right)\right|-\left|E\left(N(v)\cap A,v\right)\right|.$$
	
	As $\left(A\cap N(v)\right)\subseteq V$, and $v\notin \left(A\cap N(v)\right)$, using $(\star\star)$ we have that 
	\begin{equation*}
	\begin{split}
	&f((A\cap N(v)) + v)-f(A\cap N(v))\\ &= |E(v,V\setminus(A\cap N(v)+v))|-|E(A\cap N(v),v)|\\
	 &=\left|E\left(v,N(v) \cap \left(V\setminus(A+v)\right)\right)\right|-\left|E\left(\left(N(v)\cap A\right),v\right)\right|,
	\end{split}
	\end{equation*}
	which proves the lemma.\qed
	%
	%
	%
\end{proof}

We now prove that making the decision to join $S$ or $\bar S$ in parallel for independent sets does not affect the approximation ratio.

\begin{lemma}\label{linearity of marginal profits}
	For every $1\leq i \leq (\Delta +1)$, it holds that $f(\mathit{OPT}_{i-1})-f(\mathit{OPT}_i) \leq [f(X_i)-f(X_{i-1})] + [f(Y_i)-f(Y_{i-1})]$.
\end{lemma}	

\begin{proof}
	Let $I=\{v_1,v_2,...,v_m\}$ be an independent set of color $i$. We show that iteration $i$ of the distributed algorithm is equivalent to $m$ iterations of the sequential one.
	
	We can simulate the $i$th iteration of the distributed algorithm as $m$ sequential iterations, where in the $j$th iteration, vertex $v_j$ makes the exact same decision it makes in the distributed algorithm. Let $X_{i-1}^j,Y_{i-1}^j,\mathit{OPT}_{i-1}^j$ represent the state of $X_{i-1},Y_{i-1}$ and $\mathit{OPT}_{i-1}$ after the $j$th iteration of the simulation. Using the above notation, we prove the lemma by showing that: 
	$$\sum_{j=1}^{m}[f(\mathit{OPT}_{i-1}^{j-1})-f(\mathit{OPT}_{i-1}^j)] \leq \sum_{j=1}^{m}[f(X_{i-1}^j)-f(X_{i-1}^{j-1})] + \sum_{j=1}^{m}[f(Y_{i-1}^j)-f(Y_{i-1}^{j-1})].$$
	
	For this, it suffices to show that $[f(\mathit{OPT}_{i-1}^{j-1})-f(\mathit{OPT}_{i-1}^j)] \leq[f(X_{i-1}^j)-f(X_{i-1}^{j-1})] + [f(Y_{i-1}^j)-f(Y_{i-1}^{j-1})]$, for all $1\leq j\leq m$.
	
	Since $I$ is an independent set, it holds that $X_{i-1}(v)\cap I = \emptyset$ and $Y_{i-1}(v)\cap I = \emptyset$ for every $v\in I$, i.e. the decision of every vertex $v\in I$ does not depend on the decisions made by the other vertices in $I$. By Lemma \ref{local marginal profit}, $f(X_{i-1} + v)-f(X_{i-1})=f(X_{i-1}(v) + v)-f(X_{i-1}(v))$, and $f(Y_{i-1}-v)-f(Y_{i-1})=f(Y_{i-1}(v)-v)-f(Y_{i-1}(v))$ for every vertex $v\in V$. Therefore, given $X_{i-1}^{j-1}, Y_{i-1}^{j-1}$ and $v_j$, an iteration of the sequential algorithm is equivalent to the $j$th iteration of the simulation. We now complete the proof using Lemma \ref{lemma 2-usm}.\qed
\end{proof}

\begin{theorem}\label{Theorem - 1/3 deterministic Max-Dicut}
	The algorithm Distributed Greedy Max-Dicut gives a 1/3-approximation for the Max-Dicut problem in $\tilde{O}(\Delta + \log^*n)$ communication rounds in the $\mathcal{CONGEST}$ model.
\end{theorem}

\begin{proof} We start by showing the approximation ratio. By Lemma \ref{linearity of marginal profits}, $f(\mathit{OPT}_{i-1})-f(\mathit{OPT}_i) \leq [f(X_i)-f(X_{i-1})] + [f(Y_i)-f(Y_{i-1})]$ for all $1\leq i \leq (\Delta +1)$. Summing up the inequality for every $i$ gives:
	$$\sum_{i=1}^{\Delta + 1}[f(\mathit{OPT}_{i-1})-f(\mathit{OPT}_i)] \leq \sum_{i=1}^{\Delta + 1}[f(X_i)-f(X_{i-1})] + \sum_{i=1}^{\Delta + 1}[f(Y_i)-f(Y_{i-1})].$$
	
	As we saw in the sequential case, this is a telescopic sum that after cancellation gives:
	$f(\mathit{OPT}_0)-f(\mathit{OPT}_{\Delta + 1}) \leq [f(X_{\Delta +1})-f(X_{0})] + [f(Y_{\Delta +1})-f(Y_{0})]=f(X_{\Delta +1})+f(Y_{\Delta +1}).$
	
	The last equality follows from the fact that in the case of Max-Dicut $f(X_0)=f(Y_0)=0$.	
	Hence the output $f(\mathit{OPT}_{\Delta + 1}) \geq f(\mathit{OPT})/3$.
	
	We now analyze the number of communication rounds needed. Coloring the graph using the algorithm of \cite{barenboim2015deterministic} takes $\tilde{O}(\Delta^{3/4}+\log^* n)$ communication rounds. After the coloring, the algorithm runs for ($\Delta + 1$) iterations, each one takes O(1) communication rounds, and hence, the time complexity follows.\qed
\end{proof}

\subsection{A Randomized 1/2-Approximation Algorithm for Max-Dicut}
As shown in \cite{buchbinder2015tight}, using random decisions improves the approximation ratio. The randomized algorithm differs from the deterministic algorithm in the decision making process. Rather than taking the most profitable decision, the algorithm takes each of the possible decisions with a probability proportional to its value. 
A formal description of the algorithm is given in Algorithm \ref{Alg - rundomized 1/2 Max-Dicut}. The variables $X_i,Y_i$ and $\mathit{OPT}_i$ are defined as in the deterministic distributed algorithm.
\negA
\begin{algorithm}[H]
	\caption{Distributed Randomized Max-Dicut}
	\label{Alg - rundomized 1/2 Max-Dicut}
	\begin{algorithmic}[1]
		\State $X_0(v)=\emptyset, Y_0(v)=N(v)$
		\State color the graph using $(\Delta +1)$ colors
		\For{i=1 to $(\Delta +1)$}
		\If{$C(v)=i$} 
		\State $a_i \leftarrow f(X_{i-1}(v)+v)-f(X_{i-1}(v))$
		\State $b_i \leftarrow f(Y_{i-1}(v)-v)-f(Y_{i-1}(v))$
		\State{$a_i'\leftarrow max\{a_i,0\},b_i'\leftarrow max\{b_i,0\}$}
		\State{\textbf{with probability} $a_i'/(a_i'+b_i')$ \textbf{do}}
		\State{\indent State($v$) = $S$}
		\State{\indent send 1}
		\State{\textbf{else} (with probability $b_i'/(a_i'+b_i')$) \textbf{do}}
		\State{\indent State($v$) = $\bar S$}
		\State{\indent send 0}
		\EndIf
		\For{each vertex $v$ in $V$ in parallel}
		\State{$X_i(v) \leftarrow X_{i-1}(v) + \{u|u\in N(v)$ $\land$ $u$ sent 1$\}$}
		\State{$Y_i(v) \leftarrow Y_{i-1}(v) - \{u|u\in N(v)$ $\land$ $u$ sent 0$\}$}
		\EndFor
		\EndFor
		\State output State
	\end{algorithmic}
\end{algorithm}
\negA

We first show the equivalent of Lemma \ref{linearity of marginal profits}, and then prove our main theorem for our algorithm.

\begin{lemma}\label{expactation}
	For every $1\leq i \leq (\Delta +1)$, it holds that: 
	\begin{equation}
	\tag{$\star$}\mathbb{E}[f(\mathit{OPT}_{i-1})-f(\mathit{OPT}_i)] \leq \frac{1}{2} \mathbb{E}[f(X_i)-f(X_{i-1}) + f(Y_i)-f(Y_{i-1})].
	\end{equation}
\end{lemma}

\begin{proof}
	As shown in \cite{buchbinder2015tight}, for the sequential randomized algorithm, where in each step only one vertex makes a decision, it holds that $\mathbb{E}[f(\mathit{OPT}_{i-1})-f(\mathit{OPT}_i)] \leq \frac{1}{2} \mathbb{E}[f(X_i)-f(X_{i-1}) + f(Y_i)-f(Y_{i-1})]$. Also, as shown in the proof for Lemma \ref{linearity of marginal profits}, denoting by $I_i$ the independent set of vertices colored with $i$, the $i$-th iteration of the distributed algorithm can be simulated by $|I_i|$ iterations of the sequential algorithm. Since the inequality holds for one iteration of the sequential algorithms, it holds for $|I_i|$ iterations, and therefore holds for the distributed algorithm.\qed
\end{proof}

\begin{theorem}
	The algorithm Distributed Randomized Max-Dicut outputs a 1/2-approximation for Max-Dicut in $\tilde{O}(\Delta + \log^*n)$ communication rounds.
\end{theorem}

\begin{proof}
	The proof is very similar to the proof of Theorem \ref{Theorem - 1/3 deterministic Max-Dicut}. Using $(\star)$, and taking a summation over all $1 \leq i \leq {\Delta +1}$, we have
	$$\sum_{i=1}^{\Delta + 1}\mathbb{E}[f(\mathit{OPT}_{i-1})-f(\mathit{OPT}_i)] \leq \frac{1}{2} \sum_{i=1}^{\Delta + 1}\mathbb{E}[f(X_i)-f(X_{i-1}) + f(Y_i)-f(Y_{i-1})].$$
	Noting that the sum is telescopic, most of the terms cancel out, and we have
	\begin{equation*}
	\begin{split}
	\mathbb{E}[f(\mathit{OPT}_{0})-f(\mathit{OPT}_{\Delta+1})] &\leq \frac{1}{2} \mathbb{E}[f(X_{\Delta+1})-f(X_{0}) + f(Y_{\Delta+1})-f(Y_{0})] \\ &\leq\frac{1}{2} \mathbb{E}[f(X_{\Delta+1})+f(Y_{\Delta+1})].
	\end{split}
	\end{equation*}
	Therefore, since $\mathit{OPT}_0=\mathit{OPT}$, we have that the output satisfies $f(X_{\Delta+1})=f(Y_{\Delta+1})=f(\mathit{OPT}_{\Delta+1}) \geq f(\mathit{OPT})/2$.
	The time complexity analysis is identical to the one for the deterministic algorithm.\qed
\end{proof}

\negA
\mysection{A Deterministic $\mathcal{LOCAL}$ Algorithm}
Our coloring-based algorithms  
may become inefficient for high degree graphs, due to the strong dependence on $\Delta$. 
Consider a clique in this model. The above algorithms require a linear number of communication rounds, while learning the entire graph and finding an optimal solution requires only $O(1)$ communication rounds in the $\mathcal{LOCAL}$ model. Indeed, there is a tradeoff between the graph diameter and the average degree of its vertices. Based on this tradeoff, we propose a faster, two-step, deterministic algorithm for Max-Cut that requires $min\{\tilde{O}(\Delta +\log^*n),O(\sqrt{n})\}$ communication rounds in the 
$\mathcal{LOCAL}$ model. The pseudocode is given in Algorithm \ref{alg - local fast}.

We call a vertex $v$ a \emph{low-degree} vertex, if $deg(v)<\sqrt{n}$, and a \emph{high-degree} vertex, if $deg(v)\geq \sqrt{n}$. Define $G_{low}$, and $G_{high}$ as the graphs induced by the low-degree vertices and the high-degree vertices, respectively. The idea is to solve the problem separately for $G_{low}$ and for $G_{high}$.

In the first step, the algorithm deletes every high-degree vertex, if there are any, and its adjacent edges, creating $G_{low}$. The deletion means that the low-degree vertices ignore the edges that connect them to high-degree vertices, and do not communicate over them. Then, the algorithm approximates the Max-Cut on $G_{low}$, using one of the coloring-based algorithms described in Section \ref{Section - coloring based algorithms}.

In the second step, the problem is solved optimally within each connected component in $G_{high}$. However, the high-degree vertices are allowed to communicate over edges which are not in $G_{high}$.
As we prove next, the distance in the original graph $G$ between any two vertices which are connected in $G_{high}$ is bounded from above by $O(\sqrt{n})$. Hence, the number of rounds needed for this part of the algorithm is bounded as well by $O(\sqrt{n})$. 
\begin{algorithm}
	\caption{Fast Distributed Greedy Max-Cut}         
	\label{alg - local fast}       
	\begin{algorithmic}[1]
		\State run Distributed Greedy Max-Cut on $G_{low}$
		\For {each connected component in $G_{high}$}
		\State learn the component topology in $G$, including all its adjacent edges
		\State let the vertex with the lowest id find an optimal solution, and determine the output for each vertex in its component
		\EndFor
		\State output the vertices decisions
	\end{algorithmic}
\end{algorithm}
\negA
\negA
\begin{lemma} \label{limited distance}
	Assume $u,v$ are connected in $G_{high}$, then the distance between $u$ and $v$ in the original graph $G$ is at most $3\sqrt{n}$.
\end{lemma}
\begin{proof} 
	Let $dist_G(v_1,v_2)$ denote the distance between the vertices $v_1,v_2\in V$ in the original graph $G$. Let $u,v$ be two connected vertices in $G_{high}$, and assume, toward a contradiction, that $dist_G(u,v)>3\sqrt{n}$.
	Let $\{A_i\}_{i=0}^{m}=(u=a_0,a_1,...,a_m=v)$ be a sequence of vertices that lie on a shortest path from $u$ to $v$ in $G_{high}$. For each pair of vertices $(a_i,a_{i+1}),i=0,..,m-1$ on the path, it holds that $|dist_G(a_i,u)-dist_G(a_{i+1},u)|\leq 1$. Therefore, there is a subsequence $\{A_{i_j}\}_{j=0}^{k}$ for $k>3\sqrt{n}$, which starts with $u$ and ends with $v$, such that for every $j=0,..,k-1$ it holds that $dist_G(a_{i_{j+1}},u)-dist_G(a_{i_j},u)=1$.
	
	Note that if $a_{i_{j_1}}$ and $a_{i_{j_2}}$ have a common neighbor, then $|dist_G(a_{i_{j_1}},u)-dist_G(a_{i_{j_2}},u)|<3$.
	Since there are at least $\frac{k}{3}$ vertices in $\{A_{i_j}\}_{j=0}^{k}$, such that the distance between them is at least $3$, and each of them is of degree at least $\sqrt{n}$, we have that the number of vertices in $G$ is at least $\frac{k}{3}\cdot\sqrt{n}>n$. This contradicts the assumption that $dist_G(u,v)>3\sqrt{n}$.\qed
\end{proof}
\begin{theorem}
	Fast Distributed Greedy Max-Cut yields a $\frac{1}{2}$-approximation to Max-Cut, using $min\{\tilde{O}(\Delta +\log^*n),O(\sqrt{n})\}$ communication rounds in the $\mathcal{LOCAL}$ model.
\end{theorem}

\begin{proof}
	We first prove the approximation ratio. Since Distributed Greedy Max-Cut is applied on $G_{low}$, at least half of the edges of $G_{low}$ are in the cut.
	Given the decisions of vertices in $G_{low}$, the algorithm finds an optimal solution for all vertices in $G_{high}$. Note that running Distributed Greedy Max-Cut on the high-degree vertices of $G$, would give at least half of the remaining edges. This is due to the fact that the algorithm makes sequential greedy decisions. Therefore, an optimal solution for the high-degree vertices guarantees at least half of the edges in $G\setminus G_{low}$, implying the approximation ratio.
	
	Applying Distributed Greedy Max-Cut on $G_{low}$ requires $\tilde{O}(\Delta_{low} +\log^*n)$ communication rounds, where $\Delta_{low}=min\{\Delta,\sqrt{n}\}$. Using Lemma \ref{limited distance} we have that each high degree vertex can communicate with every high-degree vertex connected to it in $G_{high}$, using at most $O(\sqrt{n})$ communication rounds. Hence, Steps $2.-4.$ of the algorithm take $O(\sqrt{n})$ communication rounds. We note that when $\Delta<\sqrt{n}$, the algorithm terminates after the first step. Thus,
	the algorithm requires $min\{\tilde{O}(\Delta +\log^*n),O(\sqrt{n})\}$ communication rounds.\qed
\end{proof}

Using the above technique, we obtain a fast, deterministic algorithm for the Max-Dicut problem, by replacing the call to Distributed Greedy Max-Cut in Step $1.$ with a call to Distributed Greedy Max-Dicut.
Using the same arguments as in the analysis for the Max-Cut algorithm, we have:

\begin{theorem}
	Fast Distributed Greedy Max-Dicut yields a $\frac{1}{3}$-approximation to Max-Dicut, using $min\{\tilde{O}(\Delta +\log^*n),O(\sqrt{n})\}$ communication rounds in the $\mathcal{LOCAL}$ model.
\end{theorem}

\comment{
\mysection{Discussion}
In this paper we addressed Max-Cut in the classic distributed network models, where the graph represents a synchronous communication network. Our clustering-based and coloring-based techniques led to the development of the first distributed approximation algorithms for Max-Cut and Max-Dicut on general graphs in these models. We mention below some avenues for future work. It would be interesting to understand how close is the complexity of our algorithms to the true complexity of the problems in each of the distributed models. Moreover, the gap between the best approximation ratio achieved in the $\mathcal{CONGEST}$ model and the best approximation ratio achieved in the $\mathcal{LOCAL}$ model raises the immediate question of how congestion affects the ability to approximate the problems.
Another direction is to obtain lower bounds in terms of running times in approximating these problems. In particular: what is the lower bound for achieving an approximation factor strictly better than $1/2$ in the $\mathcal{CONGEST}$ model?
Finally, it would be interesting to extend the algorithms for other, and perhaps more general submodular maximization problems.
}
\myparagraph{Acknowledgements}
We thank Roy Schwartz and Shay Kutten for stimulating discussions and for helpful comments on the paper.


\bibliographystyle{splncs03}
\bibliography{MaxCut}

\end{document}